\documentclass[12pt,a4paper,english]{article}

\usepackage{fullpage}
\usepackage{color}
\usepackage{latexsym,amssymb,amsthm,amsmath}

\usepackage{graphicx}

\usepackage[cp850]{inputenc}

\usepackage{enumerate} 

\usepackage{cite} 

\newtheorem{theorem}{Theorem}
\newtheorem{proposition}{Proposition}
\newtheorem{lemma}{Lemma}
\newtheorem{corollary}{Corollary}

\title{On $k$-Convex Polygons\footnote{Partially supported by the FWF
    Joint Research Program `Industrial Geometry' S9205-N12, Projects
    MEC MTM2006-01267, DURSI 2005SGR00692, Project Gen. Cat. DGR 2009SGR1040,
    DGR 2009SGR-1040, MEC MTM2009-07242, MEC MTM2008-04699-C03-02,
    and the bilateral Spain-Austria program `Acciones Integradas' ES
    01/2008 and HU2007-0017.}}

\author{O. Aichholzer\thanks{Institute for Software Technology,
        University of Technology, Graz, Austria,
        {\tt oaich@ist.tugraz.at}}
        \and
        F. Aurenhammer\thanks{Institute for Theoretical Computer Science,
        University of Technology, Graz, Austria,
        {\tt auren@igi.tugraz.at}}
        \and
        E. D. Demaine\thanks{Computer Science and Artificial Intelligence Laboratory,
        Massachusetts Institute of Technology, Cambridge, USA,
        {\tt edemaine@mit.edu}}
        \and
        F. Hurtado\thanks{Departament de Matem\`atica Aplicada II,
        Universitat Polit\`ecnica de Catalunya, Barcelona, Spain,
        {\tt Ferran.Hurtado@upc.edu}}
        \and
        P. Ramos\thanks{Departamento de Matem\'aticas, Universidad de Alcal\'a,
        Madrid, Spain, {\tt pedro.ramos@uah.es}}
        \and
        J. Urrutia\thanks{Instituto de Matem\'aticas,
        Universidad Nacional Aut\`onoma de M\`exico,
        {\tt urrutia@matem.unam.mx}}
      }

\begin{document}
\maketitle


\begin{abstract}
  We introduce a notion of $k$-convexity and explore
  polygons in the plane that have this property. Polygons which are \mbox{$k$-convex} can be
  triangulated with fast yet simple algorithms. However, recognizing them in general is a
  3SUM-hard problem. We give a characterization of
  \mbox{$2$-convex} polygons, a particularly interesting class, and show how
  to recognize them in \mbox{$O(n \log n)$}
  time. A description of their shape is given as well, which leads to Erd\H{o}s-Szekeres
  type results regarding subconfigurations of their vertex sets. Finally, we
  introduce the concept of generalized geometric permutations, and show that their number
  can be exponential in the number of \mbox{$2$-convex} objects considered.
\end{abstract}

\section{Introduction}
\label{intro}

The notion of convexity is central in geometry. As such, it has been generalized in many ways and for different reasons. In this paper we consider a simple and intuitive generalization of convexity, which to the best of our knowledge has not been worked on.  It leads to an appealing class of polygons in the plane with interesting structural and algorithmic properties.

\medskip

A set in $\mathbb{R}^d$ is convex if its intersection with every straight line is connected. This definition may be relaxed to directional convexity or \mbox{$D$-\emph{convexity}}~\cite{FW,MP},
by considering only lines parallel to one out of a (possibly infinite) set~$D$ of vectors.  A~special case is
\emph{ortho-convexity}~\cite{RW}, where only horizontal and vertical lines are allowed. For any fixed~$D$, the family of \mbox{$D$-convex} sets is closed under intersection, and thus can be treated in a systematic way using the notion of \emph{semi-convex} spaces~\cite{SW}, which is sometimes appropriate for investigating visibility issues. The \emph{$D$-convex hull} of a set~$M$ is the
intersection of all \mbox{$D$-convex} sets that contain~$M$. If~$D$ is a finite set, this definition of a convex hull may lead to an undesirably sparse structure---an effect which can be remedied by using a stronger, functional (rather than set-theoretic) concept of \mbox{$D$-convexity}~\cite{MP}.

\medskip

\paragraph*{\boldmath{$k$}-Convex Sets}  We consider a different
generalization of convexity. A~set~$M$ in~$\mathbb{R}^d$ is called $k$-\emph{convex (with respect to transversal lines)} if there exists no straight line that intersects~$M$ in more than $k$ connected components. Throughout the paper we will use the term \mbox{$k$-\emph{convex}}, for short\footnote{We face notational ambiguity. The term
  \mbox{`$k$-convex'} has, maybe not surprisingly, been used in   different settings, namely, for functions~\cite{P}, for
  graphs~\cite{ADS}, and for discrete point sets~\cite{pavel,KL}. Also, the concept of \mbox{$k$-point} convexity~\cite{Va} has later been called \mbox{$k$-convexity} in~\cite{BK}.}.
Note that \mbox{$1$-convexity} refers to convexity in its standard meaning. To reformulate in terms
of visibility, call two points \mbox{$x,y \in M$}
$k$-\emph{visible} if \mbox{$\overline{xy} \cap M$} consists of at most $k$ components.  Thus a set is \mbox{$k$-convex} if and only if any two of its points are mutually \mbox{$k$-visible}. Applications of this concept may arise from placement problems for modems that have the capacity of transmitting through a fixed number of walls~\cite{AFFHHU,UFR}.

\begin{figure}[t]
\begin{center}
\includegraphics[width=6cm]{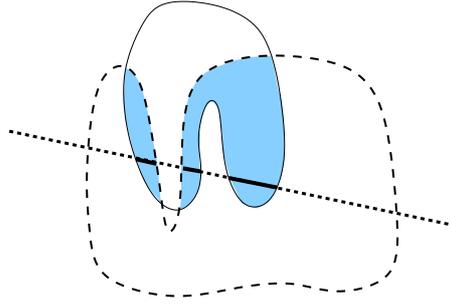}
\end{center}
\caption{Intersecting two $2$-convex sets}
\label{notkconvex}
\end{figure}

\medskip

Unlike directional convexity, \mbox{$k$-convexity} fails to show the intersection property: The intersection of \mbox{$k$-convex} sets is not \mbox{$k$-convex} in general (for fixed~$k$). Figure~\ref{notkconvex} gives an example. For \mbox{$k \geq 2$}, a \mbox{$k$-convex} set~$M$ may be disconnected, or if connected, its boundary may be disconnected. In this paper, we will
restrict attention (with an exception in Section~\ref{polydomain}) to simply connected sets in two dimensions,
namely, simple polygons in the plane.

\medskip

There are two notions of planar convexity that appear to be close to ours.  One is \mbox{$k$-\emph{point} convexity}~\cite{Va,BK} which requires that for any $k$ points in a set~$M$ in~$\mathbb{R}^2$, at least one of the line segments they span is contained in~$M$. Thus
\mbox{$2$-point} convex sets are precisely the convex sets.  The other is \mbox{$k$-\emph{link convexity}}~\cite{MSD}, being fulfilled for a given polygon~$P$ if, for any two points in~$P$, the geodesic path connecting them inside~$P$ consists of at most $k$ edges. The \mbox{$1$-link} convex polygons are just the convex polygons.  While there is a relation between \mbox{$k$-convexity} and the former concept (as we will show in Section~\ref{kconvex}), the latter concept is totally unrelated.

\medskip

We will study basic properties of \mbox{$k$-convex} polygons, in
comparison to existing polygon classes and convexity concepts in Section~\ref{kconvex}. This offers an alternative to the approach in~\cite{ABDGLS} to define `realistic' polygons as those being guardable (visible) by at most $k$ guards.  We prove that given a simple polygon~$P$, the problem of finding the smallest $k$ such that $P$ is  \mbox{$k$-convex} (equivalently, to find the stabbing number of~$P$) is 3SUM-hard.
On the other hand, a recognition algorithm that runs in $O(n^2)$ time for a polygon
with $n$ vertices is easy to obtain. Interestingly, \mbox{$k$-convex} polygons can be triangulated, by a quite simple method, in $O(n \log k)$ time. An $O(nk)$ time complexity is
achieved in~\cite{ABDGLS} for \mbox{$k$-guardable} polygons.

\medskip

The first nontrivial value, \mbox{$k=2$}, deserves particular attention. Already in this case,
a novel class of polygons is obtained. A characterization of
\mbox{$2$-convex} polygons is given in Section~\ref{2convex}. It leads  to an $O(n \log n)$ time algorithm for recognizing such polygons. Note that \mbox{$2$-convex} polygons add to the list of special classes of polygons~\cite{ET,ABDGLS} that allow for simple $O(n)$ time
triangulation methods. We also provide a qualitative description of their shape,
which implies an Erd\H{o}s-Szekeres type result, namely, that every \mbox{$2$-convex}
polygon with $n$ vertices contains a subset of at least $\sqrt n$ vertices in convex position, and that
its vertex set can be decomposed into at most $2\sqrt{2n}$ subsets in convex position.

\medskip

In Section~\ref{polydomain}, we turn our attention to general \mbox{$k$-convex} domains.
We give observations on the union and intersection properties of such domains,
and elaborate on an attempt to
generalize the notion of geometric permutations from convex sets to \mbox{$k$-convex} sets.
In contrast to the $O(m)$ bound in~\cite{ed-sh} on the number of
geometric permutations of $m$ convex sets, it turns out that the number of
generalized geometric permutations
can be exponential in $m$, already for \mbox{$2$-convex} sets. For \mbox{$2$-convex}
\textit{polygonal} domains, the number of generalized geometric permutations is $O(n^2)$,
if $n$ denotes the total number of their vertices.

\medskip

Various open questions are raised by the proposed concept of \mbox{$k$-convexity}. We list
those which seem most interesting to us, along with a brief discussion of our
results, in Section~\ref{discussion}.

\section{\mbox{$k$-Convex Polygons}}
\label{kconvex}

\subsection{Basic properties}
\label{kgeneralities}

We start with exploring some basic properties of \mbox{$k$-convex} polygons, and
compare them to existing polygon classes and related concepts.
All geometric objects we will consider are closed sets in the Euclidean plane.

\medskip

Let~$P$ be a simple polygon, and denote by~$n$ the number of vertices of~$P$.  Two line segments $e$ and~$e'$ are said to \emph{cross} if \mbox{$e \cap e'$} is a point in the relative interior of both $e$ and~$e'$.  Clearly, a polygon~$P$ is \mbox{$k$-convex} if every line segment with endpoints in~$P$ crosses at most \mbox{$2(k-1)$} edges of~$P$.  The \emph{stabbing
  number}~\cite{W,FLM} of a set of (interior-disjoint) line segments is the largest possible number of intersections between the set and a straight 
line. A~polygon is \mbox{$k$-convex} if and only if its stabbing number is at most~\mbox{$2k$}. Thus, all our observations on \mbox{$k$-convexity} could be reformulated in terms of stabbing numbers. Moreover, we have the following property, which we state for future reference.

\begin{lemma}\label{l:subset}
If $P$ is a $k$-convex polygon, then any polygon living on a subset of vertices of $P$ is also $k$-convex.
\end{lemma}

\begin{figure}[h]
\begin{center}
\includegraphics[width=14cm]{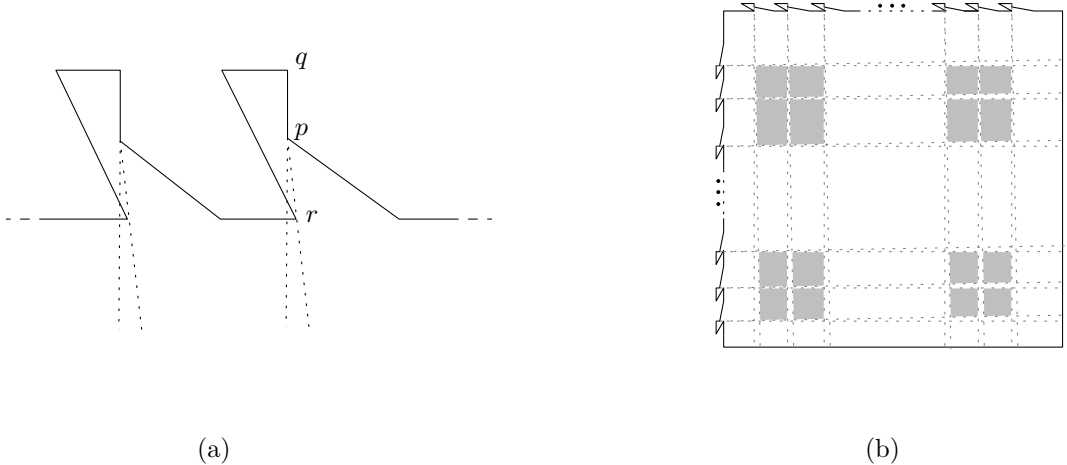}
\end{center}
\caption{Quadratic $2$-kernel construction}
\label{kernel}
\end{figure}

The kernel of a simple polygon~$P$ is the set of points that see all the polygon. Its generalization to $k$-convexity shows that \mbox{$2$-convexity} is already significantly more complex than standard convexity. The {\em $k$-kernel of $P$},
denoted as $M_k (P)$, is the set of points from which the entire polygon~$P$ is \mbox{$k$-visible}.
Note that~$P$ is \mbox{$k$-convex} if and only if $P=M_k$.
While $M_1$ is known to be a convex set which is computable in
$O(n)$ time~\cite{LP}, $M_2$ may have $\Omega(n^2)$ complexity: If we consider the `spike' in Figure~\ref{kernel}(a), the wedge between the lines $pq$ and $pr$ is not part of~$M_2$. If we arrange such spikes along the boundary of a rectangle, as in Figure~\ref{kernel}(b), we get
a quadratic number of disconnected areas which are part of the \mbox{$2$-kernel}. Therefore, any algorithm computing $M_2$, or trying to check 2-convexity via the comparison of $P$ and $M_2$ would have $\Omega(n^2)$ complexity in the worst case.

\begin{figure}[h]
\begin{center}
\includegraphics[width=10cm]{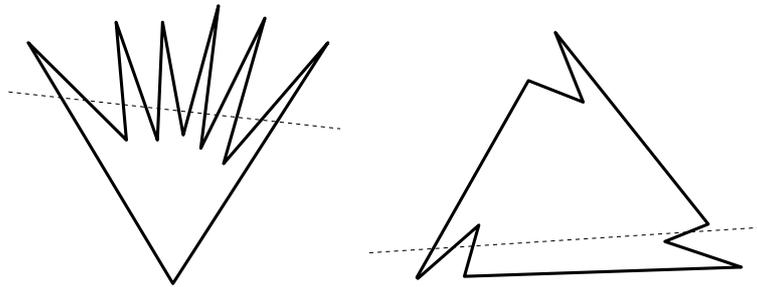}
\end{center}
\caption{Star-shaped versus 2-convex}
\label{starshaped}
\end{figure}

There is also no immediate relation to \emph{star-shaped} polygons, i.e., polygons~$P$ with~$M_1 \neq
  \emptyset$. Figure~\ref{starshaped} shows a polygon on the left hand side which is star-shaped but only
\mbox{$\frac{n}{2}$-convex}.  On the right hand side, we see a
polygon which is \mbox{$2$-convex} but not star-shaped.  Visually, \mbox{$2$-convexity} seems to be closer to convexity than is star-shapedness.  Note that cutting a \mbox{$2$-convex} polygon with any straight line leaves (at most) three parts, each being \mbox{$2$-convex} itself. This is not true in general
for star-shaped polygons.

\medskip

While \mbox{$2$-convexity} obviously restricts the winding
number of a polygon~\cite{V}, its link distance~\cite{MSD} is
unaffected and may well be~$\Theta(n)$.  Conversely, a polygon which is \mbox{$2$-link} convex (such that any two of its points are at link distance~$2$ or less) may fail to be \mbox{$k$-convex} for sublinear~$k$. The star-shaped polygon in Figure~\ref{starshaped} (left) is an example.

\begin{figure}[t]
\begin{center}
\includegraphics[scale=1.5]{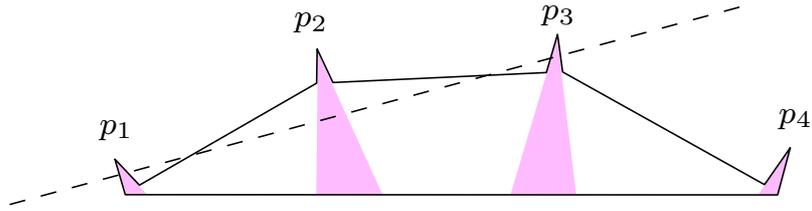}
\end{center}
\caption{Guarding a $3$-convex polygon} \label{guard2}
\end{figure}

\medskip

There is, interestingly, a relation to \mbox{$k$-{\emph{point}}} convexity as defined in~\cite{Va}.
A polygon $P$ is called $k$-point convex if for any $k$ points
$p_1, \ldots , p_k$ in $P$, at least one of the closed segments $p_i p_j$ belongs to $P$. Every \mbox{$k$-point} convex
polygon~$P$ is \mbox{$(k-1)$-convex}.  To verify this,
we prove that if~$P$ is not \mbox{$(k-1)$-convex}, then $P$ is not \mbox{$k$-point} convex: Because ~$P$ is not \mbox{$(k-1)$-convex}, there exists a line~$\ell$ which intersects~$P$ in at least $k$ components.
If we select a point in each component, it is clear that none of the segments defined by them is inside~$P$, and therefore~$P$ is not \mbox{$k$-point} convex.
However, no implication exists in the other direction.  For example, the
\mbox{$2$-convex} polygon in Figure~\ref{f:n3} fails to
be \mbox{$k$-point} convex for $k<\frac{n}{3}$. Also, any \mbox{$k$-point} convex polygon can be expressed as the union of $m$ convex polygons, where $m$ depends (exponentially) on~$k$ but is independent of the polygon size~$n$; see~\cite{BK}. Such a property is not shared by \mbox{$k$-convex} polygons, as can be seen from the \mbox{$2$-convex} polygon in Figure~\ref{f:n3}.

\medskip

\begin{figure}[htbp!]
\begin{center}
\includegraphics[width=0.7\textwidth]{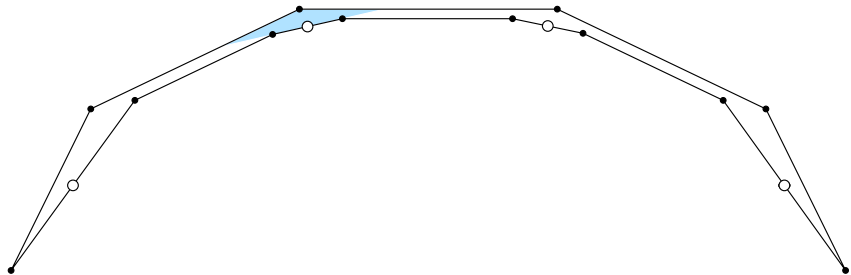}
\end{center}
\caption{Every white dot requires a different guard} \label{f:n3}
\end{figure}

The class of \mbox{$k$-convex} polygons also differs from the class of
\mbox{$k$-\textit{guardable}} polygons defined in~\cite{ABDGLS}.
It is known that any simple polygon with $n$ vertices can be
guarded with at most $\lfloor \frac{n}{3} \rfloor$ guards~\cite{OR,URR}.
The example in Figure~\ref{f:n3} shows that this number of guards can be already necessary for $2$-convex polygons. This is one of the reasons why most tools developed to study guarding problems of polygons are not very useful in the
study of modem illumination problems~\cite{AFFHHU}.

\medskip

\textit{Pseudo-triangles} are polygons with exactly three convex vertices, joined by three reflex side chains. Any pseudo-triangle is \mbox{$2$-convex}: If a straight line crosses a side chain twice, then it can cross each of the remaining two side chains at most once
(Figure~\ref{pseudopartition}, left).  That is, the stabbing number of a pseudo-triangle is four or less.  In the same way as a triangulation defines a partition of the underlying domain into convex polygons, any pseudo-triangulation~\cite{RSS} or any pseudo-convex decomposition~\cite{AHKST} gives a partition into \mbox{$2$-convex} polygons. It is an open problem (see Problem(e)
in the final section) whether it is possible to subdivide a polygon with $n$ vertices into
a sublinear number of \mbox{$2$-convex} polygons. If Steiner points are disallowed, then a
\mbox{$2$-convex} partition may have to consist of $\Theta(n)$ parts;
see Figure~\ref{pseudopartition} (right).

\begin{figure}[h]
\begin{center}
\includegraphics[width=10cm]{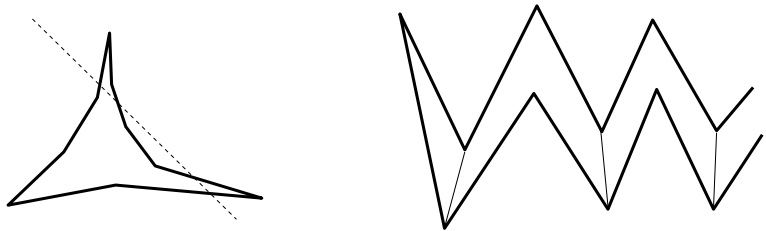}
\end{center}
\caption{Pseudo-triangle and $2$-convex partition}
\label{pseudopartition}
\end{figure}

\medskip

Two other natural questions for $k$-convex $n$-gons are decomposing them into few convex pieces in terms of~$k$, as well as giving a bound on the number of pieces of the convex hull minus the polygon. However, in both cases the answer is unrelated to $k$, because the polygon in Figure~\ref{f:2q} (left) is $2$-convex yet requires $n-2$ convex pieces, which is obviously tight because every polygon can be triangulated.
On the other hand, the polygon in Figure~\ref{f:2q} (right) is also 2-convex and has $\frac{n}{2}$ pockets, which is tight because two consecutive points on the convex hull are at least two edges apart if they are not adjacent on the polygon boundary.

\begin{figure}
\begin{center}
\includegraphics{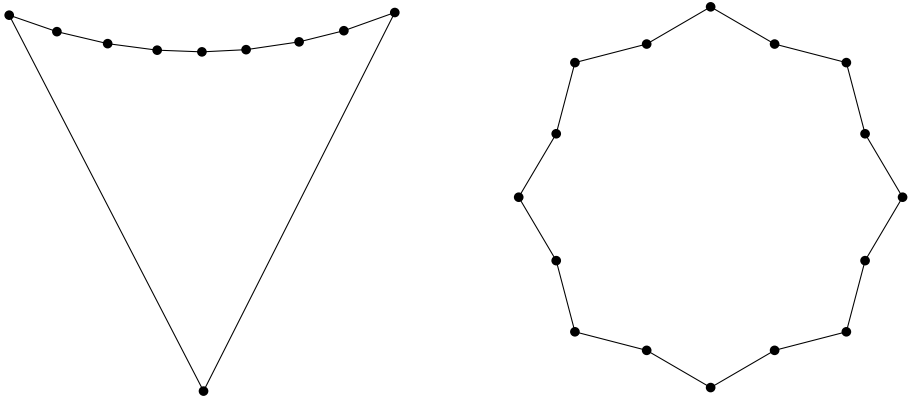}
\end{center}
\caption{A polygon with many convex pieces (left) and many pockets (right).}
\label{f:2q}
\end{figure}

\subsection{Recognition complexity}
\label{krecognition}

Let us start our considerations by observing that the stabbing number
of a polygon with $n$ vertices can be easily found in $O(n^2)$ time, as follows.
The standard duality transform maps each edge of the polygon to a double wedge consisting of lines through a common point, not including the vertical; the two lines that bound this dual wedge correspond to the endpoints of the primal segment. In the primal, any point inside a wedge is a line that stabs the segment. Therefore, the primal line that would stab most segments is in the dual a point that belongs
to as many double wedges as possible---a maximum depth problem that can easily be solved by constructing the arrangement and then traversing its cells.

The obtained $O(n^2)$ time bound is essentially tight, as in this section we prove that finding the stabbing number of a polygon, or, equivalently, finding the smallest~$k$
for which the polygon is~$k$-convex, is a 3SUM-hard problem. This family of problems is widely believed to have an $\Omega(n^2)$ lower bound for the worst case runtime \cite{GO,Ki}. We start by giving the
following result, which follows directly
from Theorem~4.1 in~\cite{GO}.


\begin{lemma}
\label{lem:3collinear}  For every integer $a$, let us consider the point $p_a=(a,a^3)$ on the cubic $y=x^3$. Then if $a,b,c$ are distinct integers, $p_a, p_b, p_c$ are collinear if and only if $a+b+c=0$.
\end{lemma}

\begin{proof}
  The points $p_a, p_b, p_c$ are collinear if and only if the
  determinant
$$
\left| \begin{array}{ccc} a & a^3 & 1\\ b & b^3 & 1\\ c &c^3 &1 \\ \end{array} \right|=
(b-a)(c-a)(c-b)(a+b+c)
$$
vanishes, which, the numbers being
different, happens exactly when $a+b+c=0$.
\end{proof}

We next show that the points $p_x$ with $x\in \mathbb{Z}$ on the cubic $y=x^3$ can be replaced by infinitesimally small vertical segments $S_x$ with upper endpoint $p_x$, such that three of them can be stabbed by a single line if and only if their three upper endpoints are collinear.

\begin{lemma}
\label{lem:hair} Let $m,a,b,c,M$ be five integers such that
$m<a<b<c<M$. Let $\varepsilon=\frac{1}{6(M-m)}$ and let $s_t$ be the (vertical) segment with endpoints $p_t=(t,t^3)$ and
$p'_t=(t,t^3-\varepsilon)$. Then $s_a, s_b$ and $s_c$ can be
stabbed by a single line if and only if $p_a$, $p_b$, and $p_c$ are collinear.
\end{lemma}

\begin{proof}
  Assume that the points $p_a, p_b, p_c$ are not collinear, and let us take three points  $q_a=(a,a^3-\varepsilon_1), q_b=(b,b^3-\varepsilon_2),  q_c=(c,c^3-\varepsilon_3)$, with $0\le\varepsilon_1, \varepsilon_2, \varepsilon_3\le\varepsilon$, i.e., three points on the segments $s_a, s_b$ and $s_c$.
  The points $q_a, q_b, q_c$ would be collinear if and only if the   determinant
\[\left| \begin{array}{ccc} a & a^3-\varepsilon_1 & 1\\ b & b^3-\varepsilon_2 & 1\\ c &c^3-\varepsilon_3 &1 \\ \end{array} \right|=
\left| \begin{array}{ccc} a & a^3 & 1\\ b & b^3 & 1\\ c &c^3 &1 \\
\end{array} \right|+\left| \begin{array}{ccc} a & -\varepsilon_1 & 1\\ b & -\varepsilon_2 & 1\\ c &-\varepsilon_3 &1 \\ \end{array}
\right|=\]
\[=\underbrace{(b-a)(c-a)(c-b)(a+b+c)}_z+\underbrace{\varepsilon_1
(b-c)-\varepsilon_2 (a-c)+\varepsilon_1 (a-b)}_\delta\] is 0,
but this is impossible because by Lemma \ref{lem:3collinear}, $z$ is an integer different from 0, which cannot become 0 by the addition of $\delta$, because

\[|\varepsilon_1 (b-c)-\varepsilon_2 (a-c)+\varepsilon_1 (a-b)|\le
|\varepsilon_1 (b-c)|+|\varepsilon_2 (a-c)|+|\varepsilon_1
(a-b)|\le \varepsilon\cdot 3(M-m)=\frac{1}{2}.\]
\end{proof}


\begin{theorem}
\label{thm:3SUMhardness} The problem of finding the stabbing
number of a polygon is 3SUM-hard.
\end{theorem}

\begin{proof}
The 3SUM problem is defined as follows: Given a set $S$ of $n$
integers, do there exist three elements $a, b, c\in S$ such that $a + b + c = 0?$ We will prove below that this problem can be reduced in $O(n\log n)$ time to the problem of computing the stabbing number of an $n$-gon. In other words, using the notation in \cite{GO},
\[\text{3SUM}\lll_{O(n\log n)} \text{stabbing number of a polygon.}\]

Let $x_1,\dots,x_n$ be the input integers. We proceed to the
reduction by steps.\medskip

\begin{figure}[htbp!]
\begin{center}
\includegraphics{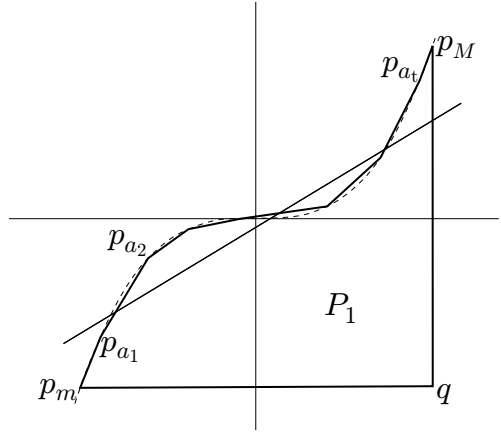}
\end{center}
\caption{Polygon $P_1$ with vertices on the cubic $y=x^3$
(dashed). The scale of the axis is not $1\colon \! 1$, to make the figure
visible.} \label{fig:polygon1}
\end{figure}

\noindent{\sl Step 1}. Sort the input numbers; let $y_1\le y_2\le \cdots \le y_n$ be the resulting list, $\mathcal{L}$. This step is done in $O(n\log n)$ time.
\medskip

\noindent{\sl Step 2}. If $0$ appears thrice in $\mathcal{L}$,
exit with a sum of three numbers in the input being $0$, otherwise continue. The step is completed in linear time. \medskip

\noindent{\sl Step 3}. If $a\ne0$ appears at least twice in
$\mathcal{L}$, check whether $-2a\in \mathcal{L}$. If
so, exit with a sum of three numbers in the input being
$0$, otherwise continue. The step is completed in $O(n\log n)$
time, as binary search can be used in the sorted list (in fact, $O(n)$ time is sufficient by scanning the list from left to right and from right to left, in a coordinated simultaneous advance).
\medskip

\noindent{\sl Step 4}. Remove multiples from the list
$\mathcal{L}$ so that each number appears exactly once. This
requires linear time. Let $a_1<a_2<\cdots <a_t$ (where $t\le n$) be these numbers.
\medskip

\noindent{\sl Step 5}. Define $m=a_1-1$, $M=a_t+1$, and $q=(M,m^3)$. Now let us consider the polygon $P_1$ whose vertices, described clockwise, are $p_mp_{a_1}p_{a_2} \dots p_{a_t}p_Mq$, where $p_x=(x,x^3)$, as in the preceding lemma
(Figure~\ref{fig:polygon1}). Observe that the stabbing number of $P_1$ is $4$, and that the polygon can be constructed in $O(n)$ time.
\medskip

\begin{figure}[htbp!]
\begin{center}
\includegraphics{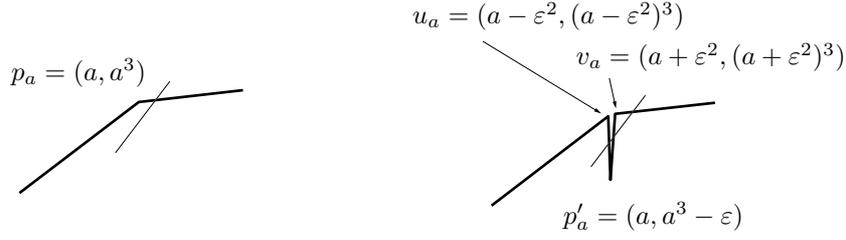}
\end{center}
\caption{Polygon $P_2$ with vertices on the cubic $y=x^3$ replaced
by infinitesimal vertical slots.} \label{fig:polygon2}
\end{figure}

\noindent{\sl Step 6}. Next  we modify the polygon $P_1$ to become the polygon $P_2$ whose vertices, described clockwise, are $p_mu_{a_1}p'_{a_1}v_{a_1}u_{a_2}p'_{a_2}v_{a_2}\dots
u_{a_t}p'_{a_t}v_{a_t}p_Mq$, where $p'_x=(x,x^3-\varepsilon),
u_x=(x-\varepsilon^2,(x-\varepsilon^2)^3),
v_x=(x+\varepsilon^2,(x+\varepsilon^2)^3))$; see
Figure~\ref{fig:polygon2}. This polygon can be constructed in
$O(n)$ time, and in the vicinity of the point $p_{a_i}$ its
stabbing number changes locally from 1 to 3; see
Figure~\ref{fig:polygon3}.
Let $s_{a_i}$ be the segment joining $u_{a_i}$ to $p'_{a_i}$.
Therefore, the stabbing number of
$P_2$ is $12$ if and only if three of the segments
$s_{a_i}$ can be simultaneously stabbed, and $8$ otherwise, as two of those segments can always be stabbed.
\medskip

\begin{figure}[htbp!]
\begin{center}
\includegraphics{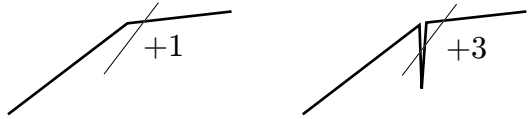}
\end{center}
\caption{From $P_1$ to $P_2$ the stabbing number locally increases
from $1$ to $3$.} \label{fig:polygon3}
\end{figure}

\noindent{\sl Step 7}. Compute the stabbing number of polygon
$P_2$ using the best possible algorithm.\medskip

\noindent{\sl Step 8}. If the stabbing number of polygon $P_2$ is $10$, conclude that there were three numbers in the initial input such that their sum is $0$; if the stabbing number of $P_2$ is $8$, conclude that there are no such three numbers.
\medskip

The correctness of Step 8 is a direct consequence of
Lemma~\ref{lem:3collinear} and Lemma~\ref{lem:hair}. As all steps but Step 7 have overall complexity $O(n\log n)$, we conclude that Step 7,  the computation of the stabbing number of a polygon, is a 3SUM-hard problem, as claimed.
\end{proof}

As an immediate consequence we obtain:

\begin{corollary}
\label{cor:kConvexity} The problem of deciding the smallest number $k$ such that a given polygon is $k$-convex is 3SUM-hard.
\end{corollary}

Also, if we replace Step 7 in the proof of
Theorem~\ref{thm:3SUMhardness} by checking whether or not polygon $P_2$ is $4$-convex, we get:

\begin{corollary}
\label{cor:4Convexity} The problem of deciding whether a given
polygon is $4$-convex is 3SUM-hard.
\end{corollary}

%

\subsection{Fast triangulation}
\label{ktriangulation}

Triangulating a simple polygon faster than in $O(n \log n)$ time with a simple method
is a challenging open problem. For \mbox{$k$-convex} polygons, this can be achieved,
by the fact that we can sort the vertices of a \mbox{$k$-convex} polygon~$P$ in any given direction (say, \mbox{$x$-direction}) in $O(kn)$ time: Simply scan around~$\partial P$ and use insertion sort, starting each time from the place where the \mbox{$x$-value} of the previous vertex has been inserted. Then any fixed value~$x_j$, once being inserted, takes part in later comparisons at most \mbox{$2k-1$} times because otherwise, the vertical line~$x=x_j$ would intersect~$P$ in more than $k$ components. Having \mbox{$x$-sorted} $P$'s vertices, a simplified plane sweep method can be used to build a vertical trapezoidation~\cite{CI,FM} (and then a triangulation) of~$P$. Only trivial data structures may be used,
as the scenario on the sweep line is of complexity~$O(k)$,
by the \mbox{$k$-convexity} of~$P$. Thus, each vertex of~$P$ can be processed in~$O(k)$ time during the sweep. We conclude:

\begin{proposition}
\label{p:triang} Any \mbox{$k$-convex} polygon can be triangulated in $O(kn)$ time and $O(n)$ space.
\end{proposition}

Using suitable data structures, a faster yet still
implementable algorithm is possible, as we show next.
Call a polygon $k$-\emph{monotone} if every vertical line intersects
the interior of the polygon in $k$ intervals.
(This property is implied by $k$-convexity, and is equivalent to $x$-monotonicity for $k=1$.)
Actually, we do not even need the polygon to be simple, we just need a sequence of $x$-coordinates such that every other $x$-coordinate comes between at most $2k$ consecutive pairs of $x$-coordinates.

\begin{lemma}
\label{lem:sort}
The vertices of any $k$-monotone polygon can be $x$-sorted in $O(n \log (2+k))$ time.
\end{lemma}

\begin{proof}
  We use a binary insertion sort, in which we add the points in order   along the polygon into a balanced binary search tree.  The binary   search tree has the \emph{dynamic finger property}: inserting an   element that has rank $r$ different from the previously inserted   element costs only $O(\log (2+r))$ time.  (For example, splay trees~\cite{ST85}
  and Brown \& Tarjan finger trees~\cite{BT80} both have this property.)  Once the elements are inserted into the binary search tree, we simply perform a linear-time in-order traversal to extract them in sorted order.
\medskip

Now we find the bound on the total cost of the insertions.  When we insert an   element of rank difference $r$ from the previously inserted element,   we can charge this cost to $r$ points formed from projecting all   vertices onto all edges of the polygon.  There are at most $O(n k)$  such points of projection, so the total number of charges is at most
  $O(n k)$.  Thus the total insertion cost is $O(\sum_{i=1}^n \log (2+r_i))$ where $\sum_{i=1}^n r_i = O(n k)$.  Such a sum is  maximized when the $r_i$ are all roughly equal, which means that they are all $\Theta((n k) / n) = \Theta(k)$.  Therefore the total cost is at most $O(n \log (2+k))$.
\end{proof}

To see that this bound is optimal in the comparison model,
consider the case in which the polygon is a comb with $k$ tines. Then sorting the $x$-coordinates is equivalent to merging $k$ sorted sequences of length $n/k$, which is known to take $\Theta(n \log k)$ time in the worst case.

\medskip

Lemma~\ref{lem:sort} yields a fast triangulation method for general \mbox{$k$-convex} polygons. We first sort the vertices of the \mbox{$k$-convex} polygon~$P$ in a fixed direction, in $O(n \log k)$ time. Again, a plane sweep is used to compute a triangulation of~$P$. As the intersection of~$P$ with the sweep line is of complexity~$O(k)$ only, by the \mbox{$k$-convexity} of~$P$, each of the $n$ vertices of~$P$ can be processed in~$O(\log k)$ time during the sweep.

\begin{theorem}
  \label{t:triang} Any \mbox{$k$-convex} polygon can be triangulated in
  $O(n \log k)$ time and $O(n)$ space.
\end{theorem}

\section{Two-Convex Polygons}
\label{2convex}

\subsection{Characterization}
\label{2character}

In this section we give a characterization of \mbox{$2$-convex} polygons that allows their recognition in time \mbox{$O(n \log
n)$}, and a description of their structure that will be used later in several of our results.

\medskip

We observe that  \mbox{$k$-convexity} is a property
that may be lost by small perturbations on the positions
of the vertices of a polygon. For example, small changes in the positions of the vertices $p_1, \ldots, p_4$
in Figure~\ref{guard2}, could yield a $1$- or $2$-convex polygon. As a consequence, stabbing a polygon along its edges will not, in most cases, give enough information for deciding
its \mbox{$k$-convexity}.

\begin{figure}[htbp!]
\begin{center}
\includegraphics[width=7.5cm]{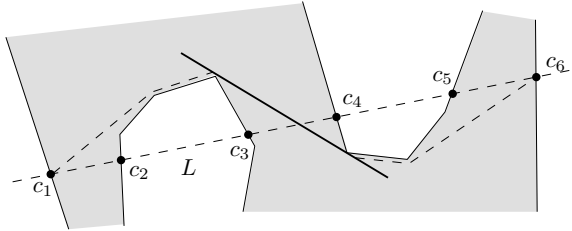}
\end{center}
\caption{Constructing an inner tangent} \label{stabber1}
\end{figure}

\medskip

Let~$P$ be a simple polygon, and denote its boundary by $\partial P$.
A line $L$ is called a \mbox{$j$-\emph{stabber}}
of $P$ if $L$ crosses $\partial P$ at least~$j$ times.  Note that a \mbox{$j$-stabber} may contain whole edges of $P$; these are \emph{not} considered to contribute to the count. An
edge of a polygon $P$ is called an \emph{inflection edge} if it joins a convex and a reflex vertex of~$P$. An \emph{inner tangent} of~$P$ is a line segment~\mbox{$T \subset P$} that contains two non-adjacent reflex vertices of $P$ in its relative interior; see Figure~\ref{stabber1}.

\begin{lemma}
\label{char}
A simple polygon~$P$ is \mbox{$2$-convex} if and only if~$P$ has no inner tangent, and no \mbox{$3$-stabber} that contains an inflection edge.
\end{lemma}

\begin{proof}
  Consider the `only if' implication first. Suppose that for some inflection edge~$e$ of~$P$, its supporting line~$L$ is a  \mbox{$3$-stabber} for~$P$. Then a slight perturbation of~$L$ gives rise to a line $L'$ that crosses $\partial P$ at least two more times than $L$ does. One of these extra crossings lies in the interior  of $e$, and one on the second edge of $P$ incident to the reflex vertex of $e$.
  Thus~$P$ admits a \mbox{$6$-stabber}, and therefore $P$
  cannot be \mbox{$2$-convex}.  Similarly, if~$P$ has some inner tangent~$T$, defined by reflex vertices~$u$ and~$v$, say, then a \mbox{$6$-stabber} cutting off~$u$ and~$v$ exists in the vicinity of~$T$.

\begin{figure}[htbp!]
\begin{center}
\includegraphics[width=8.2cm]{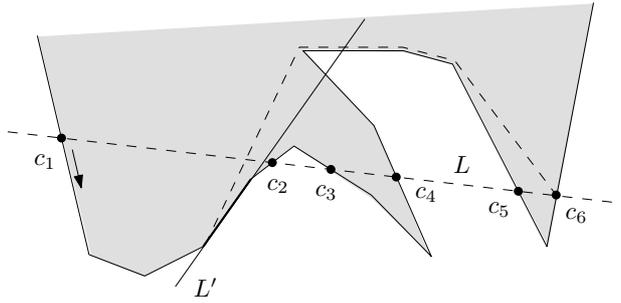}
\end{center}
\caption{Constructing a $3$-stabber}
\label{stabber2}
\end{figure}

To prove the `if' implication, assume that~$P$ is not
\mbox{$2$-convex}.  Then there exists a \mbox{$6$-stabber}
$L$ of $P$. Let \mbox{$c_1,\ldots,c_6$} denote six
consecutive crossings of~$L$ with~$\partial P$ such that the line segment joining $c_1$ to $c_2$ is contained in $P$. Two types of crossing pattern arise, as shown in Figures~\ref{stabber1} and~\ref{stabber2}, respectively. In the former case, we simply connect the two crossings~$c_1$ and~$c_6$ with a geodesic path inside~$P$ (dashed), and obtain an inner tangent (bold). In the latter case, if an inner tangent exists locally, then we can construct it in a similar way, possibly by moving the geodesic's endpoints from~$c_1$ and~$c_6$ closer to each other along~$\partial
P$. Otherwise (as Figure~\ref{stabber2} illustrates) we can find an inflection edge (bold) between~$c_1$ and~$c_3$ whose supporting line~$L'$ crosses~$P$'s boundary twice between $c_4$ and~$c_5$. But the line~$L'$~has to cross~$\partial P$ once more, and thus it is a \mbox{$3$-stabber} containing an inflection edge, which is not possible.
\end{proof}

\subsection{Recognition}
\label{2recognition}

Suppose that we want to decide if a polygon $P$ is
\mbox{$2$-convex}. (Assume that $P$ is not convex; the problem is trivial, otherwise.)
Our recognition algorithm is based on
Lemma~\ref{char}.  We look for inner tangents and \mbox{$3$-stabbers} at inflection edges.
To this end, we first shoot two rays at each reflex vertex $v$
of $P$ in the directions determined by the edges of $P$
incident with $v$. This can be done in $O(n \log n)$~\cite{CEGGHSS}.
If $\partial P$ is intersected more than once by any
of these rays, then a \mbox{$6$-stabber} exists, and we report
that~$P$ is not \mbox{$2$-convex}.

Suppose then that each ray shot at the reflex vertices
of $P$ yields a \emph{unique} intersection point with~$\partial P$. We store the points of intersection and use them to check for inner tangents. Define, for each reflex vertex~$v$ of~$P$, its \emph{critical range} $C(v)$ as the set of all points~\mbox{$x   \in \partial P$} such that $\overline{xv}$ can be prolonged to a line segment tangent to~$P$ at~$v$. Note that such a segment need not lie entirely in~$P$. However, $C(v)$ consists of exactly two connected intervals on~$\partial P$ whose endpoints are among the stored points obtained from ray shooting. See Figure~\ref{criticalrange}, where $C(v)$ is drawn with bold lines.

\begin{figure}[h]
\begin{center}
\includegraphics[width=6.5cm]{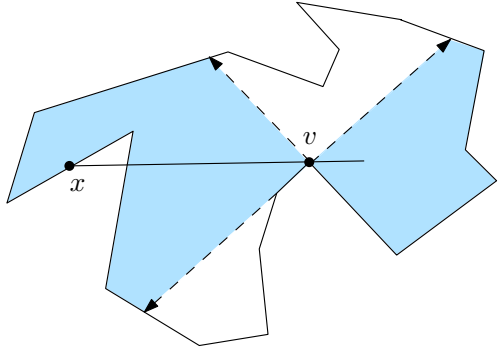}
\end{center}
\caption{Critical range for vertex~$v$}
\label{criticalrange}
\end{figure}

\begin{lemma}
  $P$ admits an inner tangent if and only if $P$ has two reflex   vertices~$v$ and~$v'$ such that \mbox{$v \in C(v')$} and \mbox{$v'\in C(v)$}.
\end{lemma}

\begin{proof}
  Clearly, if two vertices~$v$ and~$v'$ define an inner tangent   for~$P$, then \mbox{$v \in C(v')$} and \mbox{$v' \in C(v)$} hold.   Conversely, assume that both inclusion conditions are   fulfilled. Consider the geodesic path $\pi$ between~$v$ and~$v'$ inside~$P$. Either $\pi$ is the line segment~$\overline{vv'}$, and thus can be extended to an inner tangent at~$v$ and~$v'$, or $\pi$ detours via reflex vertices~$w_1,\ldots,w_m$. In the latter case (which is illustrated in Figure~\ref{detour}), we can extend the
  line segments~$\overline{vw_1}$ and~$\overline{w_mv'}$ to inner tangents, because ray shooting at reflex vertices of~$P$ has led to unique intersection points with~$\partial P$, and thus the path~$\pi$ does not cross the four (dashed) rays depicted in the figure.
\end{proof}

\begin{figure}[h]
\begin{center}
\includegraphics[width=6cm]{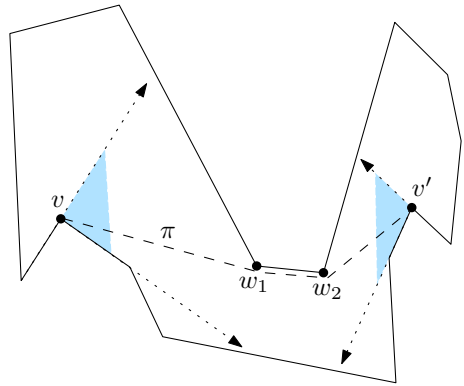}
\end{center}
\caption{Geodesic path between~$v$ and~$v'$}
\label{detour}
\end{figure}

The strategy for detecting inner tangents is now clear. First we associate to each reflex vertex $v$ two intervals, each containing a set of consecutive vertices of $P$ lying
within the critical range of $v$.  We then choose a point $x$ in $\partial P$ and calculate the set $R(x)$ of reflex vertices  containing $x$ in their critical ranges. We then slide $x$ along $\partial P$, maintaining the set~$R(x)$ in a range search tree. $R(x)$ has to be updated each time~$x$ passes over a point of intersection of $\partial P$ with a ray
shot from a reflex vertex of $P$. Moreover,
when $x$ reaches some reflex vertex~$v$ of~$P$, we
check in logarithmic time whether the intervals previously
associated to $v$ contain any element in $R(v)$.
Thus this phase takes $O(n \log n)$ time.

Since range search trees can be implemented in linear space, we have the following result:

\begin{theorem}
\label{decide}
Deciding if a simple polygon $P$ with $n$ vertices is
\mbox{$2$-convex} can be done in \mbox{$O(n \log n)$} time and $O(n)$~space.
\end{theorem}

\subsection{Shape structure}
\label{2convexsubsets}

We have given a geometric characterization of \mbox{$2$-convex} polygons,
in Subsection~\ref{2character}.
The present subsection aims at giving a qualitative description of their shape.

\begin{lemma}
\label{shape} Let $P$ be a \mbox{$2$-convex} polygon. Let
$C=p_0p_1\dots p_t$ be the chain of vertices that connects
(counterclockwise) two consecutive vertices $p_0,p_t$ on the convex hull $CH(P)$. Then $C$ can be partitioned
into three chains $C_1=p_0p_1\dots p_r$, $C_2=p_{r+1}\dots p_s$, and $C_3=p_{s+1}\dots p_t$, for $0 \leq r \leq s < t$, such that all vertices in $C_1$ and $C_3$ are convex (in~$P$), while all vertices in $C_2$ are reflex.
\end{lemma}

\begin{proof}
If $C_2$ is empty, the lemma is obviously true, so we assume that the chain $C$ contains at least one reflex vertex. Let $p_k$ be the last point of the chain $C$ to be hit in a parallel sweep starting with the line defined by $p_0p_t$ which, without loss of generality, we assume to be horizontal. Observe that $p_k$ is necessarily a reflex vertex. Let us recall that an inflection edge is adjacent to a
reflex vertex and to a convex vertex and that, according to
Lemma~\ref{char}, 2-convex polygons do not have $3$-stabbers
containing an inflection edge. We are going to see that there is at most one inflection edge in the chain $p_0\dots p_k$, and that the same applies to the chain $p_k\dots p_t$. Let $p_ip_{i+1}$ be the first inflection edge starting from $p_0$ (see Figure~\ref{f:ameba}.a). As the directed ray $p_ip_{i+1}$ intersects the chain $p_tp_0$ at least once, the chain $p_{i+1}p_t$ is necessarily to the right of the
oriented line $p_ip_{i+1}$. Observe that $p_t$ is to the right of the oriented line $p_ip_{i+1}$, because otherwise a rotation around $p_i$ starting from $p_{i-1}p_i$ will define an inner tangent, forbidden in 2-convex polygons according to Lemma~\ref{char}. Finally, if we assume that there exists a convex vertex $p_j$ in the chain $p_{i+1}p_k$, then we have that $p_{j-1}p_j$ is an inflection edge supporting a $3$-stabber: the ray $p_jp_{j-1}$ intersects the
polygon at least once, while the ray $p_{j-1}p_j$ intersects the chain $p_kp_t$ and therefore intersects the polygon at least twice.
\end{proof}

Observe that the chains $C_1$ and $C_3$ might be singletons in some cases, and that $C_2$ might be empty. However, the generic aspect of a pocket and the shape of a 2-convex polygon are as shown in Figure~\ref{f:ameba}.b). The next result follows directly:

\begin{corollary} \label{cor3}
If the convex hull of a $2$-convex polygon $P$ has $k$ vertices, then the boundary of $P$ can be decomposed
into $2k$ convex chains.
\end{corollary}

\begin{figure}[htbp!]
\begin{center}
\includegraphics[width=\textwidth]{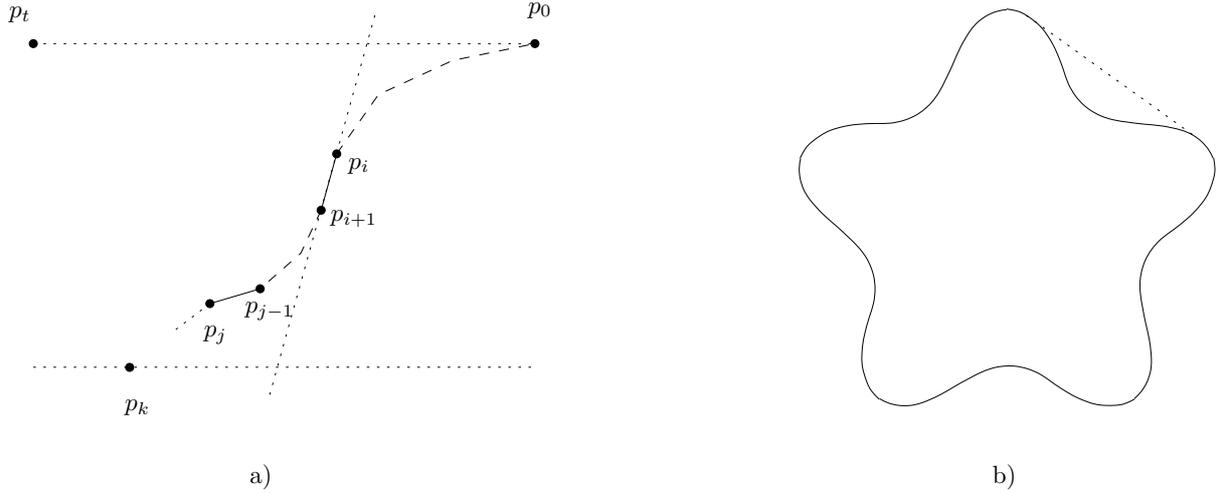}
\end{center}
\caption{Illustration for the proof of Lemma~\ref{shape}.}
\label{f:ameba}
\end{figure}

The Erd\H{o}s-Szekeres Theorem says that every set of $n$ points in general position contains at least $\log n$ points that are in convex position, and that this value is asymptotically tight~\cite{ES}. As every point set can be `polygonized', one cannot expect a better value when the points are chosen from the set of vertices of an arbitrary polygon.
However, when a point set is the set of vertices of a
\mbox{$2$-convex} polygon, we can improve this bound as
follows.

\begin{theorem}\label{thm:convexsubset}
Every \mbox{$2$-convex} polygon with $n$ vertices has a subset of $\lceil\sqrt{n/2}\rceil$
vertices in convex position. This bound is tight.
\end{theorem}

\begin{proof}
By Corollary~\ref{cor3}, the boundary of a $2$-convex polygon
with $k$ vertices on its convex hull can be decomposed into at
most $2k$ convex chains. If $k\geq\lceil\sqrt{n/2}\rceil$, we are done, otherwise one of the $2k$ convex chains necessarily has size at least $\lceil\sqrt{n/2}\rceil$. The amoeba-like example in Figure~\ref{f:ameba}.b), with $k=\lceil\sqrt{n/2}\rceil$ vertices in the convex hull and $2k$ convex chains of equal size shows that this bound is tight.
\end{proof}

We conclude this section with a consequence of the preceding
theorem.

\begin{corollary}\label{cor:convexPartition}
If an $n$-gon is \mbox{$2$-convex}, then its vertices can be
grouped into at most $2\sqrt{2n}$ subsets, each in
convex position.
\end{corollary}

\begin{proof}

Let $S(n)$ be the number of convex subsets needed to partition the vertex set of a 2-convex polygon with $n$ vertices. We show that $S(n) \leq \alpha \sqrt{n}$ by induction over $n$. The induction base for $n = 3$ is obvious, and valid for any $\alpha \geq 1$. By Theorem 4 we find one convex subset of size at least $\lceil\sqrt{n/2}\rceil$. Moreover, by Lemma~\ref{l:subset}, the remaining points define also a $2$-convex polygons and we have
$S(n) \leq 1+S(\lfloor n -
\sqrt{n/2} \rfloor) \leq 1 + \alpha \sqrt{\lfloor n - \sqrt{n/2}
\rfloor}$, where the last inequality comes from the induction
hypothesis. To prove the lemma it is sufficient to show that $1 +\alpha \sqrt{n - \sqrt{n/2}} \leq \alpha \sqrt{n}$. Standard
manipulation shows that this is true for any $\alpha \geq 2\sqrt{2}$ and any $n \geq 1$.


\end{proof}

%
%
%

\section{General $k$-Convex Domains}
\label{polydomain}

The union or intersection of simple polygons may not be a polygon.
In view of this fact, the issue of how the degree of convexity behaves with
respect to these operations is not meaningful for this class of objects.
In this section, we consider larger classes of sets in~$\mathbb{R}^2$
for which these natural questions may be discussed. We first study some properties
of general (not necessary polygonal) subsets of~$\mathbb{R}^2$.

\begin{lemma}
The union of a $k$-convex set $Q_1$ and an $m$-convex set $Q_2$ is a $(k+m)$-convex set, which is tight.
\end{lemma}

\begin{proof}
The number of intersections of a line with the boundary of $Q_1\cup Q_2$ can be at most $2k+2m$. On the other hand, if $Q_1$ and $Q_2$ are disjoint and the line that gives $k$ and $m$ connected components, respectively, is the same, the value is achieved.
\end{proof}

\begin{lemma}
The intersection of a $k$-convex set $Q_1$ and an $m$-convex set $Q_2$ is a $(k+m-1)$-convex set,
which is tight.
\end{lemma}

\begin{proof}
  An oriented line will cut the boundary of $Q_1$ at most $2k$ times and the boundary of $Q_2$ at most $2m$ times. However, the first intersection point $a$ does not contribute to the total number of cuts with $Q_1\cap Q_2$ unless $a\in Q_1\cap Q_2$, in which case it contributes only once instead of twice as intersection. The same happens with the last point, which gives the upper bound. An example proving the tightness
appears in Figure~\ref{f:int}(a).
\end{proof}

\begin{figure}
\begin{center}
\includegraphics{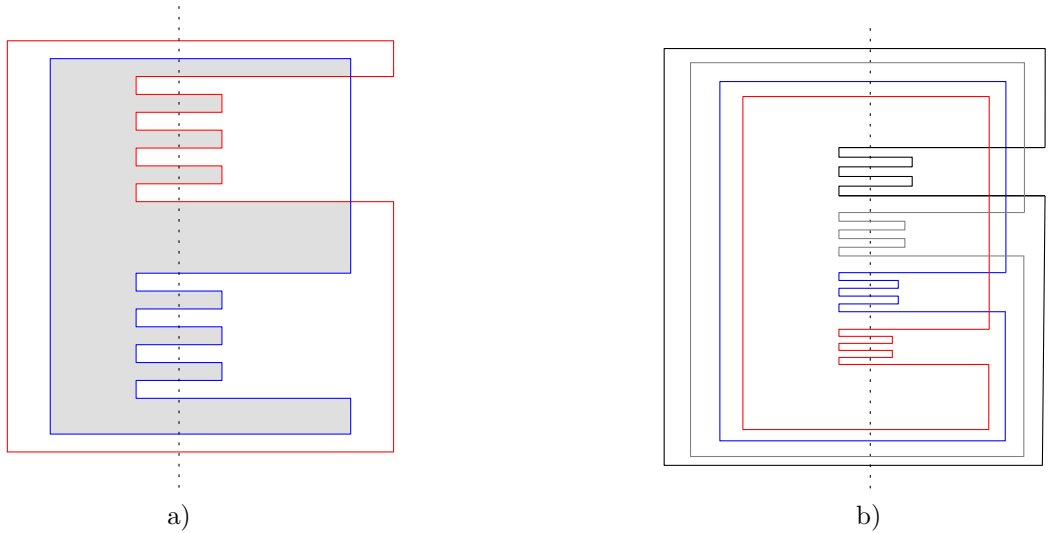}
\end{center}
\caption{Intersecting $k$-convex sets}
\label{f:int}
\end{figure}

\begin{corollary}
  The intersection of a family of $m$ $k$-convex sets is a $(m(k-1)+1)$-convex set, which is tight.
\end{corollary}

\begin{proof}
  The upper bound follows from the preceding theorem, and a construction giving its tightness is shown in Figure~\ref{f:int}(b).
\end{proof}

\begin{theorem}
  There is no Helly-type theorem for $k$-convex sets.
\end{theorem}

\begin{proof}
  We are constructing a family of $m$ 2-convex sets such that any subfamily has nonempty intersection yet there is no point common to all of them.
Let $Q_m$ be a regular polygon with $m$ edges $e_1,\ldots,e_m$ (refer to Figure~\ref{f:helly}).
Let $P_i^*$ be the polygonal chain obtained from the boundary of $Q_m$ by removing edge $e_i$ and an infinitesimal portion of $e_{i-1}$ and $e_{i+1}$. Finally, let us give some slight thickness to the chain so it becomes a polygon $P_i$. Notice that the polygons $P_1,\ldots,P_m$ are 2-convex, the intersection $\bigcap_{i=1}^m P_i$ is clearly empty, while the intersection of every proper subfamily $F$ is nonempty because it contains the intervals $e_j$ for all those $P_j\not\in F$.
\end{proof}

\begin{figure}
\begin{center}
\includegraphics{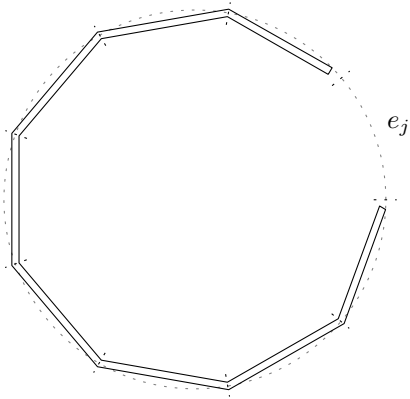}
\end{center}
\caption{No Helly-type theorem for $k$-convex sets}
\label{f:helly}
\end{figure}

The preceding lemmas apply to $k$-convex sets in general, not only polygonal domains.
(A \textit{polygonal domain} is a set obtained by a finite number of union operations of simple polygons.) However, a significant difference appears in our next results, that are possibly the most natural to explore, because they involve transversal lines, which are precisely the main concept underlying the definition of $k$-convexity.

\medskip

Let us recall \cite{wenger} that given a family of sets $Q_1,\ldots,Q_m$, a line $\ell$ is said to be a {\em transversal} of the family if $\ell$ has a nonempty intersection with each of the sets. When the sets are convex, the ordering in which they are traversed (disregarding the orientation of the line) is called a {\em geometric permutation}, a topic that has received significant attention \cite{wenger}. In particular, it has been proved that $m$ compact disjoint convex sets admit at most $2m-2$ geometric permutations, which is tight~\cite{ed-sh}.

\medskip

Let us consider now transversals of compact 2-convex sets. Notice that every object will appear at least once, but may appear twice on the transversal, which we consider as  combinatorially different cases of the associated {\em generalized geometric permutation}. What is the maximum number of generalized geometric permutations a family of $m$ 2-convex sets may have?

\begin{theorem}
The number of generalized geometric permutations of a set of $m$ $2$-convex objects
may be exponential in~$m$.
\end{theorem}

\begin{figure}
\begin{center}
\includegraphics[width=16cm]{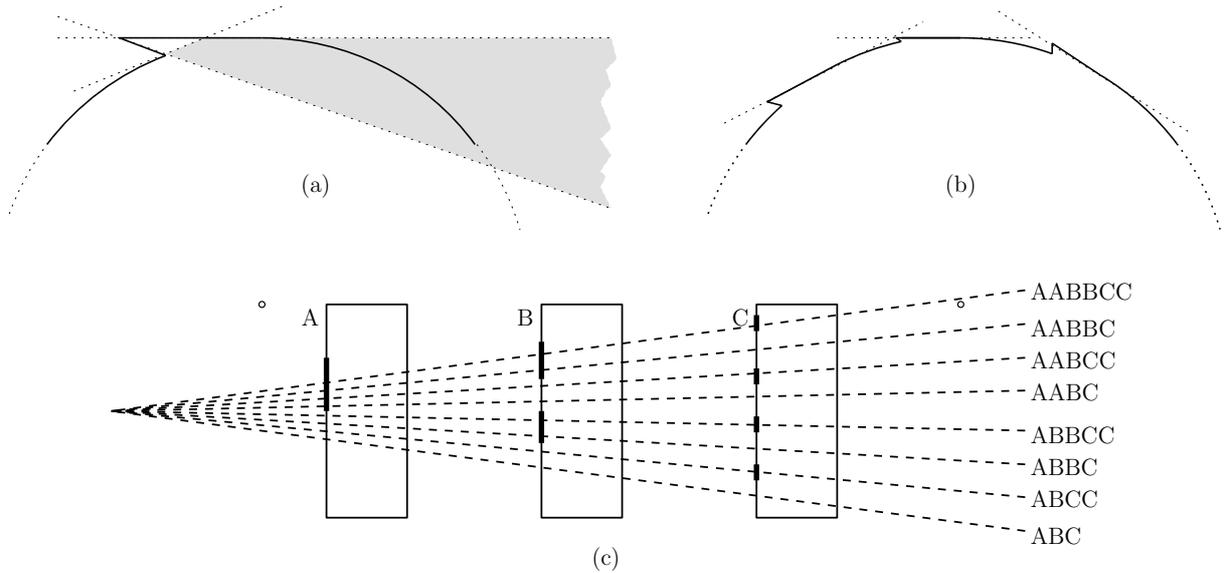}
\end{center}
\caption{(a) and (b): 2-convex objects can have an unbounded number of noses. (c): The number of generalized geometric permutations for a set of 2-convex objects can be exponential.}
\label{f:exponential}
\end{figure}

\begin{proof}
  We first show that a 2-convex objects can be arbitrarily complex.  A
  {\it nose} of an object $O$ is a zig-zag sequence of a reflex and a
  convex vertex of the boundary of $O$ as depicted in
  Figure~\ref{f:exponential}(a). Locally a line 'normal' to the nose
  intersects $O$ in two connected components. The shaded area
  in Figure~\ref{f:exponential}(a) indicates the region which is not
  intersected by a line tangent to one of the vertices of the
  nose. Thus we can iteratively construct further noses in this region
  without destroying the 2-convexity of
  $O$. Figure~\ref{f:exponential}(b) shows an example where the
  principle shape of $O$ is part of a disk. Observe that when the
  radius of the disk is large enough we can arrange an arbitrary
  number of flat noses such that $O$ stays 2-convex.

\medskip

  Let $R_i$ be an object which has the base shape of an axis-aligned
  rectangle, where the left side is actually part of a circle with
  sufficiently large radius and a center point far to the right of
  $R_i$. We place $2^{i-1}$ noses along this side, so that $R_i$ stays
  2-convex as described above. Next we arrange $k$ objects $R_1$ to $R_k$
  from left to right, as depicted in Figure~\ref{f:exponential}(c) for
  $k=3$. We position the noses for each $R_i$ in a regular way such
  that a rotating line (see the dashed lines in
  Figure~\ref{f:exponential}(c)) intersects the noses in the same
  manner as the digit ``1'' shows up in the sequence of all $2^k$ binary
  numbers of length $k$. Thus, we get $2^k$ different generalized
  geometric permutations for this setting, as each object appears
  twice if the nose is intersected, but only once otherwise.
\end{proof}

\begin{theorem}
  The maximum number of generalized geometric permutations of a set of $2$-convex polygonal domains with a total of $n$ boundary edges is $\Theta(n^2)$.
\end{theorem}

\begin{figure}
\begin{center}
\includegraphics{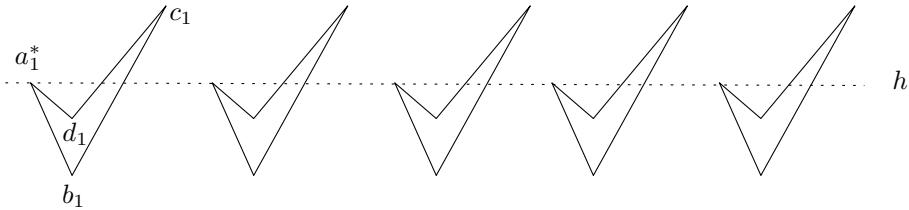}
\end{center}
\caption{The number of generalized geometric permutations for a set of 2-convex polygons can be quadratic.}
\label{f:geo_permut}
\end{figure}

\begin{proof}
  As the standard duality transform maps each edge to a double wedge, the induced arrangement of $2n$ lines in the dual plane yields a quadratic number of cells that bound from above the number of possible ways of stabbing the set of objects.

\medskip

  To see that the bound is asymptotically tight, we give a construction using $n$ 2-convex polygons; in fact, the simplest possible ones, namely, nonconvex quadrilaterals.
  Let $Q_1^*=a_1^*b_1c_1d_1$ be the quadrilateral shown in Figure~\ref{f:geo_permut}. Let $h$ be a horizontal line through $a_1^*$ and let $Q_i^*=a_i^*b_ic_id_i$ be translates of $Q_1$ in such a way that all of them are pairwise disjoint and $a_1^*,a_2^*,\ldots,a_n^*$ appear in this order on $h$. Finally, let us perturb infinitesimally $a_i^*$ to $a_i$ in such a way that
  \begin{itemize}
    \item[a)] points $a_1,a_2,\ldots,a_n$ are in general position,
    \item[b)] for all $i,j,k$, with $i\neq j$, the line $a_ia_j$ leaves above the point $c_k$ and below the point $d_k$.
  \end{itemize}
  For $i=1,\ldots,n$, let $Q_i$  be the quadrilateral with vertices $a_ib_ic_id_i$. There are $\binom{n}{2}$ lines of the type $a_ia_j$; each of them leaves above and below a different set of points $a_k$, and is a transversal because it crosses all the segments $b_kc_k$ for every $k$. Now, if $a_k$ is below the transversal, $Q_k$ is intersected once, while if $a_k$ is above the transversal, $Q_k$ is intersected twice. Therefore, we have obtained $\binom{n}{2}$ generalized geometric permutations.
\end{proof}

Observe that the two preceding theorems apply {\it mutatis mutandis} to $k$-convex sets, because 2-convex sets are also $k$-convex for $k \geq 3$.

\section{Discussion and Open Problems}
\label{discussion}

In this paper we have considered a new concept of generalized convexity. Moving from convexity to 2-convexity is seemingly a small change, as we are just accepting lines to intersect in at most two connected components instead of one. It is remarkable that this modest departure has strong consequences in the complexity of the new class of objects, as we have seen in this paper; obviously, even more when the degree of convexity is increased.
Several open problems remain and many interesting  questions can be raised. We list some of them below.
\begin{itemize}

  \item[a)] Can the recognition of \mbox{$2$-convex} polygons be carried out in linear time, improving on the $O(n\log n)$ algorithm we provide?

 \item[b)] Finding the smallest $k$ such that a given polygon is \mbox{$k$-convex} is a 3SUM-hard problem. In particular, recognizing \mbox{$4$-convexity} is already 3SUM-hard. We do not know whether the situation is the same for \mbox{$3$-convexity} or whether a subquadratic time algorithm exists for this case.

 \item[c)] Is it possible to generalize
      Theorem \ref{thm:convexsubset}? For example, is it true that every \mbox{$k$-convex} polygon with $n$ vertices has a large subset of vertices that are the vertices of a \mbox{$(k-1)$-convex} polygon?

\item[d)] Let us define the $k$-convex hull of a point set $S$ as the smallest area polygon which is $k$-convex, has a subset $T\subset S$ as vertex set and every point in $S\smallsetminus T$ is inside the polygon. Which is the complexity of computing this $k$-convex hull? Observe that for $k=1$ this notion is the usual convex hull of a point set.

 \item[e)] Give combinatorial bounds and efficient algorithms for decomposing a polygon into $k$-convex subpolygons. This is a classical problem when convex subpolygons are considered \cite{Ke} and also the decomposition into pseudotriangles, which are 2-convex polygons, has been studied \cite{RSS}. However, the latter result might be improved by considering more general 2-convex polygons.

\item[f)] A $k$-convex decomposition of a set~$S$ of $n$ points in the plane is a decomposition of its convex hull into $k$-convex polygons such that every point in $S$ is a vertex of some of the polygons. For $k=1$, a triangulation suffices, though it has been proved that if we allow arbitrary convex sets the number can be reduced~\cite{NRU}. On the other hand, it has been shown that there exist always a decomposition into exactly $n-2$ pseudotriangles, which are 2-convex polygons \cite{RSS}. It is an intringuing open problem to decide whether this number can be reduced to sublinear if we allow arbitrary 2-convex polygons.
\end{itemize}

Finally, let us mention that in this paper we have focused on
\mbox{$k$-convex} \emph{polygons}. It is natural to define a similar
concept for finite point sets, namely, being in \mbox{$k$-\emph{convex position}}, where $k$
is given by the smallest degree of convexity attained when \emph{all} possible
polygonizations of the point set are considered. This issue is
considered in a companion paper \cite{second}.

\bigskip

\noindent{\bf Acknowledgements} We would like to thank Thomas
Hackl, Clemens Huemer, David Rappaport, Vera \mbox{Sacrist\'an} and Birgit Vogtenhuber for sharing discussions on this topic.

\bibliographystyle{abbrv}

\begin{thebibliography}{99}

\bibitem{AFFHHU}  O. Aichholzer, R. Fabila-Monroy, D. Flores-Pe\~naloza, T. Hackl, C. Huemer, J. Urrutia, and B. Vogtenhuber.
\newblock{Modem Illumination of Monotone Polygons.}
\newblock{In Proc. 25th European Workshop on Computational Geometry EuroCG '09, Brussels, Belgium, 2009, 167-170.}

\bibitem{second} O. Aichholzer, F. Aurenhammer, T. Hackl, F. Hurtado, P.A. Ramos, J. Urrutia, P. Valtr,
B. Vogtenhuber.
\newblock {\em k-convex points sets.}
\newblock Manuscript, 2010.

\bibitem{AHKST}
O. Aichholzer, C. Huemer, S. Kappes, B. Speckmann, C.D. T\'oth.
\newblock {\em Decompositions, partitions, and coverings with convex polygons and pseudo-triangles.}
\newblock Graphs and Combinatorics 23 (2007), 481-507.

\bibitem{ABDGLS}
G. Aloupis, P. Bose, V. Dujmovic, C. Gray, S. Langerman, B. Speckmann.
\newblock {\em Triangulating and guarding realistic polygons.}
\newblock Proc. $20^{th}$ Canadian Conference on Computational Geometry, 2008.


\bibitem{ADS}
D. Artigas, M.C. Dourado, J.L. Szwarcfiter.
\newblock {\em Convex partitions of graphs.}
\newblock Electronic Notes in Discrete Mathematics 29 (2007), 147-151.

\bibitem{BK}
M. Breen, D.C. Kay.
\newblock {\em General decomposition theorems for $m$-convex sets in the plane.}
Israel Journal of Mathematics 24 (1976), 217-233.

\bibitem{BT80}
M.R. Brown, R.E. Tarjan.
\newblock {\em Design and Analysis of a Data Structure for Representing
                   Sorted Lists.}
SIAM Journal on Computing 9(3) (1980), 594-614.

\bibitem{CEGGHSS}
B. Chazelle, H. Edelsbrunner, M. Grigni, L. Guibas, J. Hershberger, M. Sharir, J. Snoeyink.
\newblock {\em Ray shooting in polygons using geodesic triangulations.}
\newblock Algorithmica 12 (1994), 54-68.

\bibitem{CI}
B. Chazelle, J. Incerpi.
\newblock {\em Triangulation and shape-complexity.}
\newblock ACM Transactions on Graphics 3 (1984), 135-152.

\bibitem{ed-sh}
H. Edelsbrunner, M. Sharir.
\newblock {\em The maximum number of ways to stab $n$ convex non-intersecting sets in the plane is $2n-2$}
\newblock Discrete Comput. Geom. 5 (1990), 35û42.

\bibitem{ET}
H. ElGindy, G.T. Toussaint.
\newblock {\em On geodesic properties of polygons relevant to linear time triangulation.}
The Visual Computer 5 (1989), 68-74.

\bibitem{ES}
P. Erd\H{o}s, G. Szekeres.
\newblock {\em A combinatorial problem in
geometry.} Compositio Mathematica, 2 (1935), 463-470.

\bibitem{UFR}  R. Fabila-Monroy, Andres Ruis Vargas, and J. Urrutia.
\newblock{On Modem Illumination Problems.}
\newblock{XIII Encuentros de Geometria Computacional, Zaragoza, Espa\~{n}a, June 29 - July 1 2009, 9-19.}

\bibitem{FLM}
S.P. Fekete, M.E. L\"ubbecke, H. Meijer.
\newblock {\em Minimizing the stabbing number of matchings, trees, and triangulations.}
\newblock Proc. $15^{th}$ Ann. ACM-SIAM Symp. on Discrete Algorithms, 2004, 430-439.

\bibitem{FW}
E. Fink, D. Wood.
\newblock {\em  Fundamentals of restricted-orientation convexity.}
Information Sciences 92 (1996), 175-196.

\bibitem{FM}
A. Fournier, D.Y. Montuno.
\newblock {\em Triangulating simple polygons and equivalent problems.}
ACM Transactions on Graphics 3 (1984), 153-174.

\bibitem{GO}
A. Gajentaan and M.~H. Overmars. \newblock {\em On a class of
$O(n^2)$ problems in computational geometry.} CGTA: Computational
Geometry: Theory and Applications, 5, (1995), 165-185.

\bibitem{Ke}
M. Keil.
\newblock{\em Polygon Decomposition} in Handbook of Computational Geometry (2000), 491-518.

\bibitem{Ki}
J. King. \newblock {\em A Survey of 3SUM-Hard Problems.}
http://www.cs.mcgill.ca/~jking/


\bibitem{KL}
G. Kun, G. Lippner.
\newblock {\em Large empty convex polygons in k-convex sets.}
\newblock Periodica Mathematica Hungarica 46 (2003), 81-88.

\bibitem{LP}
D.T. Lee, F.P. Preparata.
\newblock {\em An optimal algorithm for finding the kernel of a polygon.}
\newblock Journal of the ACM 26 (1979), 415-421.

\bibitem{MSD}
A. Maheshwari, J.-R. Sack, H. Djidjev.
\newblock {\em Link distance problems.}
\newblock In: Handbook of Computational Geometry, J.-R. Sack and J. Urrutia (eds.),
Elsevier, 2000, 519-558.

\bibitem{MP}
J. Matou\u{s}ec, P. Plech\'a\u{c}.
\newblock {\em On functional separately convex hulls.}
\newblock Discrete \& Computational Geometry 19 (1998), 105-130.

\bibitem{NRU}
V. Neumann-Lara, E. Rivera-Campo, J. Urrutia.
\newblock{A note on convex decompositions of point sets in the plane.}
\newblock{Graphs and Combinatorics} 20(2) (2004), 223-231.

\bibitem{OR}
J. O'Rourke.
\newblock {\em Art Gallery Theorems and Algorithms.}
\newblock Oxford University Press, 1987.

\bibitem{P}
T. Pennanen.
\newblock {\em Graph-convex mappings and $K$-convex functions.}
\newblock Journal of Convex Analysis 6 (1999), 235-266.

\bibitem{RW}
G.J.E. Rawlins, D. Wood.
\newblock {\em Ortho-convexity and its generalizations.}
\newblock In: Computational Morphology, G.T. Toussaint (ed.),
\newblock North-Holland Publishing Co.,
Amsterdam, 1988, 137-152.

\bibitem{RSS}
G. Rote, F. Santos, I. Streinu.
\newblock {\em Pseudo-triangu\-lations---a survey.}
Contemporary Mathematics 453 (2008), 343-410.

\bibitem{SW}
S. Schuierer, D. Wood.
\newblock {\em Visibility in semi-convex spaces.}
\newblock Journal of Geometry 60 (1997), 160-187.

\bibitem{ST85}
D.D. Sleator, R.E. Tarjan.
\newblock {\em Self-adjusting binary search trees.}
\newblock Journal of the ACM 32(3) (1985), 652-686.

\bibitem{URR}  J. Urrutia.
\newblock{Art Gallery and Illumination Problems.}
\newblock{Handbook on Computational Geometry, Elsevier Science Publishers, J.R. Sack and J. Urrutia eds. pp. 973-1026, 2000.}

\bibitem{Va}
F.A. Valentine.
\newblock {\em A three-point convexity property.}
\newblock Pacific Journal of Mathematics 7 (1957), 1227-1235.

\bibitem{pavel}
P. Valtr.
\newblock {\em A sufficient condition for the existence of large empty convex polygons}
\newblock Discrete Comput. Geom., 28 (2002), 671û682.

\bibitem{V}
G. Vegter.
\newblock {\em Kink-free deformations of polygons.}
\newblock Proc. $5^{th}$ Ann. ACM Symp. Computational Geometry, 1989, 61-68.

\bibitem{W}
E. Welzl.
\newblock {\em On spanning trees with low crossing numbers.}
\newblock In: Data Structures and Efficient Algorithms,
\newblock B. Monien and T. Ottmann (eds.),
Springer LNCS 594, 1992, 233-249.

\bibitem{wenger}
R. Wenger,
\newblock{Progress in geometric transversal theory.}  \newblock{Contemporary Mathematics, B. Chazelle and J.E. Goodman, Eds., American Math Society} (1999), 375-393.

\end{thebibliography}

\end{document}